\documentclass[11pt]{article}
\usepackage{amssymb}

\usepackage{lipsum}
\usepackage{amsfonts}
\usepackage{graphicx}
\usepackage{epstopdf}
\usepackage{algorithmic}
\usepackage{amsthm}
\usepackage{tabularx,ragged2e,booktabs,caption}
\usepackage{algorithm}
\usepackage{authblk}

\newtheorem{theorem}{Theorem}
\newtheorem{definition}{Definition}
\newtheorem{lemma}{Lemma}

\newtheorem{proposition}{Proposition}
\newtheorem{corollary}{Corollary}

\newcommand{\alga}{{\bf TreeIncl1}}
\newcommand{\algb}{{\bf TreeIncl2}}
\newcommand{\proca}{{\bf ComputeSet}}
\newcommand{\lnon}{\overline}

\newlength{\boxwidth}
\newcommand{\SSS}{\mathcal{S}}

\begin{document}

\title{New and Improved Algorithms for Unordered Tree Inclusion}
\author[1]{Tatsuya Akutsu
\thanks{Correspoding author. e-mail: takutsu@kuicr.kyoto-u.ac.jp.
Partially supported by JSPS KAKENHI \#18H04113.}}
\author[2,3]{Jesper Jansson}
\author[1]{Ruiming Li}
\author[4]{Atsuhiro Takasu}
\author[1]{Takeyuki Tamura\thanks{Partially supported by JSPS KAKENHI \#25730005.}}
\affil[1]{Bioinformatics Center,
Institute for Chemical Research, Kyoto University,
Uji, Kyoto 611-0011, Japan}
\affil[2]{Department of Computing,
The Hong Kong Polytechnic University, Hung Hom, Kowloon, Hong Kong, China}
\affil[3]{Graduate School of Informatics,
Kyoto University, Yoshida-Honmachi, Sakyo-ku, Kyoto 606-8501, Japan}
\affil[4]{National Institute of Informatics,
Chiyoda-ku, Tokyo, 101-8430, Japan}

\maketitle

\begin{abstract}
The \emph{tree inclusion problem} is, given two node-labeled trees~$P$ and~$T$
(the ``pattern tree'' and the ``target tree''), to locate every minimal subtree
in~$T$ (if any) that can be obtained by applying a sequence of
node insertion operations to~$P$.
Although the \emph{ordered} tree inclusion problem is solvable in
polynomial time, the \emph{unordered} tree inclusion problem is
NP-hard.
The currently fastest algorithm for the latter is a classic algorithm by
Kilpel\"{a}inen and Mannila from 1995 that runs in
$O(2^{2d} mn)$ time,
where $m$ and $n$ are the sizes of the pattern and target trees, respectively,
and $d$ is the degree of the pattern tree.
Here, we develop a new algorithm that 
runs in $O(2^{d} mn^2)$ time,
improving the exponential factor from $2^{2d}$ to~$2^d$
by considering
a particular type of ancestor-descendant relationships that is suitable for
dynamic programming.
We also study restricted variants of the unordered tree inclusion problem.

\end{abstract}

\section{Introduction}

Tree pattern matching and measuring the similarity of trees are
fundamental problem areas in theoretical computer science.
One intuitive and previously well-studied measure of the similarity between
two rooted, node-labeled trees~$T_1$ and~$T_2$ is
the \emph{tree edit distance}, defined as the length of a shortest sequence of
node insertion, node deletion, and node relabeling operations that
transforms~$T_1$ into~$T_2$~\cite{bille2005survey}.
An important special variant
of the problem of computing the tree edit distance
known as the \emph{tree inclusion problem} is obtained when the only allowed
operations are node insertion operations on~$T_1$.
Here, we assume the following formulation of the problem:
given a ``pattern tree''~$P$ and a ``target tree''~$T$, locate
every minimal subtree in~$T$ (if any) that can be obtained by applying
a sequence of node insertion operations to~$P$.
(Equivalently, one may define the tree inclusion problem so that only
node deletion operations on~$T$ are allowed.)
In~1995, Kilpel\"{a}inen and Mannila~\cite{kilpelainen1995ordered} proved that
the tree inclusion problem for unordered trees is NP-hard, but solvable in
polynomial time when the degree of the pattern tree is bounded from above by
a constant.
The running time of their algorithm is
$O(d \cdot 2^{2d} \cdot mn) = O^{\ast}(2^{2d}) = O^{\ast}(4^{d})$, where
$m = |P|$, $n = |T|$, and $d$ is the degree of~$P$.
Throughout this article, the notation ``$O^{\ast}(\dots)$'' means
``$O(\dots)$'' multiplied by some function that is a polynomial in~$m$ and~$n$.
E.g., ``$O^{\ast}(2^{2d})$'' means ``$O(2^{2d} \cdot poly(m,n))$''.

Our main contribution is a new algorithm for solving
the unordered tree inclusion problem more efficiently.
More precisely, its time complexity is
$O(d \cdot 2^{d} \cdot mn^2) = O^{\ast}(2^{d})$,
which yields the first improvement in over twenty years.
Our bound is obtained by introducing the simple yet useful concept of
\emph{minimal inclusion} and considering a particular type of
ancestor-descendant relationships that turns out to be suitable for
dynamic programming.
Next, we analyze the computational complexity of unordered tree inclusion for
some restricted cases; see Table~\ref{table:complexity} for a summary of
the new findings.
We give a polynomial-time algorithm for the case where the leaves of~$P$ are
distinctly labeled and every label appears at most twice in~$T$, and
an $O^{\ast}(1.619^d)$-time algorithm for the NP-hard case where the leaves
in~$P$ are distinctly labeled and each label appears at most three times
in~$T$.
Both of these algorithms effectively utilize some techniques from
a polynomial-time algorithm for 2-SAT~\cite{aspvall1979}.
(Note that the preliminary version of this paper \cite{akutsu2018unordered}
contained a slower algorithm for the latter case running in $O^{\ast}(1.8^d)$
time.)
Finally, we derive a randomized $O^{\ast}(1.883^d)$-time algorithm for the case
where the heights of~$P$ and~$T$ are one and two, respectively, via
a non-trivial combination of our $O^{\ast}(2^d)$-time algorithm,
Yamamoto's algorithm for SAT~\cite{yamamoto2005sat},
and color-coding~\cite{alon1995color}.

\medskip

\begin{table}[h!]
\caption{The computational complexity of some special cases of
the unordered tree inclusion problem.
For any tree~$T$, $h(T)$~denotes the height of~$T$ and
$occ(T)$ the maximum number of times that any leaf label occurs in~$T$.
As indicated in the table, either all nodes or only the leaves are labeled
(the former is harder since it generalizes the latter).
Note that the fourth case is NP-hard as it generalizes the first two.
The algorithm referred to in the last case is randomized.
}
\begin{center}
\begin{tabular}{l|c|c|c} \hline
Restriction & Labels on & Complexity & Result \\ \hline
$h(T) = 2$, $h(P) = 1$, 
  & all nodes
  & NP-hard
  & Corollary~\ref{corollary:KM_hardness} \\
$occ(T) = 3$, $occ(P) = 1$
& & & \\[2mm]
$h(T) = 2$, $h(P) = 2$, 
  & leaves
  & NP-hard
  & Theorem~\ref{thm:unique_leaves_NP-complete} \\
$occ(T) = 3$, $occ(P) = 1$
& & & \\[2mm]
$occ(T) = 2$, $occ(P) = 1$
  & all nodes
  & P
  & Theorem~\ref{thm:poly} \\[2mm]
$occ(T) = 3$, $occ(P) = 1$
  & all nodes
  & $O^{\ast}(1.619^d)$ time
  & Theorem~\ref{thm:occ-3} \\[2mm]
$h(T) = 2$, $h(P) = 1$
  & all nodes
  & $O^{\ast}(1.883^d)$ time 
  & Theorem~\ref{thm:low-height} \\
\hline
\end{tabular}
\end{center}
\label{table:complexity}
\end{table}

\medskip

\subsection{Related results}

In general, tree edit distance-related problems are computationally harder
for unordered trees than for ordered trees.
A comprehensive summary of the many results that were already known in~2005
can be found in the survey 
by Bille~\cite{bille2005survey}.
Below, we briefly mention a few of these historical results along with
some more recent ones.

When $T_1$ and~$T_2$ are \emph{ordered} trees, the tree edit distance can be
computed in polynomial time.
The first algorithm to do so, invented by Tai~\cite{tai1979tree} in~1979, ran
in $O(n^{6})$ time, where $n$ is the total number of nodes in~$T_1$ and~$T_2$.
The time complexity was gradually improved upon until
Demaine et al.~\cite{demaine2009optimal} thirty years later presented
an $O(n^{3})$-time algorithm, which was proved to be worst-case optimal under
the conjecture that there is no truly subcubic time algorithm for
the all-pairs-shortest-paths problem~\cite{bringmann2018editlb}.
Pawlik and Augsten~\cite{pawlik2011rted} developed a robust algorithm whose
asymptotic complexity is less than or equal to the complexity of
the best competitors for any input instance.
In another line of research, since even $O(n^3)$ time is too slow for
similarity search and so-called join operations in XML databases, the focus has
been on approximate methods.
Garofalakis and Kumar~\cite{garofalakis2005xml} gave an algorithm for embedding
the tree edit distance in a high-dimensional $L_1$-norm space with
a guaranteed distortion, and recently, Boroujen
et al.
provided an $O(n^2)$-time $(1+\varepsilon)$-approximation
algorithm~\cite{boroujeni2019treeedit}.

In contrast, the tree edit distance problem is NP-hard for
\emph{unordered} trees~\cite{zhang1992editing}.
It is MAX SNP-hard even for binary trees in
the unordered case~\cite{zhang1994some}, which implies that it is unlikely
to admit a polynomial-time approximation scheme.
Some exponential-time algorithms for this problem variant were developed
by Akutsu et al.~\cite{akutsu2013approximation,akutsu2014efficient}.
As for parameterized algorithms, Shasha et al.~\cite{shasha1994exact} gave
an $O(4^{\ell_1 + \ell_2} \cdot \min(\ell_1,\ell_2) \cdot mn)$-time algorithm
for the problem, where $\ell_1$ and $\ell_2$ are the numbers of leaves in~$T_1$
and~$T_2$, respectively.
Taking the tree edit distance (denoted by~$k$) to be the parameter instead,
an $O^{\ast}(2.62^k)$-time algorithm for the unit-cost edit operation model was
developed
by Akutsu et al.~\cite{akutsu2011exact}.

As mentioned above, Kilpel\"{a}inen and Mannila~\cite{kilpelainen1995ordered}
proved that the unordered tree inclusion problem is NP-hard and gave
an algorithm that runs in $O(d \cdot 2^{2d} \cdot mn)$ time, where $m = |P|$,
$n = |T|$, and $d$ is the degree of~$P$.
Bille and G{\o}rtz~\cite{bille2011tree} presented a fast algorithm for the case
of ordered trees, and Valiente~\cite{valiente2005constrained} gave
a polynomial-time algorithm for a constrained version of the unordered case.
Piernik and Morzy~\cite{piernik2013partial} introduced a similar problem for
ordered trees and developed an efficient algorithm.
Finally, we remark that the special case of the tree inclusion problem in which
node insertion operations are only allowed to insert new \emph{leaves}
corresponds to a subtree isomorphism problem, which can be solved in
polynomial time even for unordered trees~\cite{matouvsek1992complexity}.

\subsection{Applications}
\label{sec:applications}

Research in tree pattern matching has led to algorithms used in
numerous practical applications over the years.
Some examples include
fast methods for querying structured text databases,
document similarity search,
natural language processing, 
compiler optimization,
automated theorem proving,
comparison of RNA secondary structures,
assessing the accuracy of phylogenetic tree reconstruction methods,
and medical image analysis~\cite{bille2011tree,icde14,kilpelainen1995ordered,valiente2005constrained}.
Recently, due to the rapid advance of AI technology, matching methods
for knowledge bases have become increasingly important.
In particular, researchers in the database community have enhanced
the basic \emph{subtree similarity search} technique to search a knowledge base
of hierarchically structured information under various definitions of
similarity;
e.g., Cohen and Or~\cite{icde14} presented a general subtree similarity search
algorithm that is compatible with a wide range of tree distance functions, and
Chang
et al.~\cite{vldb15}
proposed a top-$k$ tree matching algorithm.
In the Natural Language Processing (NLP) field,
researchers are applying deep learning techniques to NLP problems and
developing algorithms for processing parsing/dependency trees~\cite{emnlp15}.

As an example of the versatility of tree comparison algorithms,
three different tree pattern matching applications involving glycan data
from the KEGG database~\cite{kanehisa2013data},
weblogs data~\cite{zaki2005efficiently}, and bibliographical data from
ACM, DBLP, and Google Scholar~\cite{kopcke2010evaluation} were all expressed
in terms of an optimization version of the unordered tree inclusion problem
named the \emph{extended tree inclusion problem} and studied experimentally
by Mori et al.~\cite{MTJHTA_15}.
Note that for bibliographic matching, a single article usually has at most
two or three versions (e.g., preprint, conference version, and journal version),
and it is very rare that a single article includes two co-authors with
exactly the same family and given names.
Therefore, two reasonable assumptions when modeling bibliographic matching as
the tree inclusion problem
are
that the leaves of the pattern tree~$P$ are
distinctly labeled and that each label occurs at most~$c$ times in
the target tree~$T$ for some bounded value of~$c$.

Another important restriction is on the height of trees.
In entity resolution, some authors have applied tree matching
where entities are usually represented by a shallow tree.
Mori et al. \cite{MTJHTA_15} represented a bibliographic record by a tree
of height 2 and linked identical records in two different bibliographic
databases.
Konda et al. \cite{konda2016er} evaluated their entity resolution system
by using various datasets.
Movie records from IMDb\footnote{IMDb: Ratings, Reviews, and Where to Watch
the Best Movies: https://www.imdb.com} used in their experiment,
for example, were extracted from IMDb web pages in HTML format.
The fields of the movie record are included in a subtree of height 7
in the web page.
Since the HTML code contains many tags for rendering the page,
the height of trees required for the movie record is much lower.
Apart from these practical viewpoints, 
it is of theoretical interests to study restricted cases because
the unordered tree inclusion problem remains NP-hard even in 
considerably restricted
cases,
as shown in Table~\ref{table:complexity}.

\section{Definitions and notation}
\label{sec:definitions}

From here on, all trees are assumed to be rooted, unordered, and node-labeled.
Let~$T$ be a tree.
We use $r(T)$, $h(T)$, and~$V(T)$ to denote the root of~$T$, the height of~$T$,
and the set of nodes in~$T$, respectively.
For any $v \in V(T)$, $\ell(v)$~is the node label of~$v$ and $Chd(v)$~is
the set of children of~$v$.
Furthermore, $Anc(v)$~and $Des(v)$ are the sets of strict ancestors and
strict descendants of~$v$, respectively (i.e., $v$~itself is excluded from
these sets), whereas $AncDes(v) = Anc(v) \cup Des(v) \cup \{v\}$~is the set of
all ancestors of~$v$, all descendants of~$v$, and~$v$.
Also, $T(v)$~denotes the subtree of~$T$ induced by $Des(v) \cup \{v\}$.

A \emph{node insertion operation} on a tree~$T$ is an operation that
creates a new node~$v$ having any label and then:
(i)~attaches~$v$ as a child of some node~$u$ currently in~$T$
and makes $v$
become the parent of a (possibly empty) subset of the children of~$u$
instead of~$u$, so that $u$ is no longer their parent;
 or
(ii)~makes the current root of~$T$ become a child of~$v$ and lets $v$ become
the root of $T$ instead.
For any two trees~$T_1$ and~$T_2$, we say that
\emph{$T_1$~is included in~$T_2$} if there exists a sequence of
node insertion operations that, when applied to~$T_1$, yields~$T_2$.
Equivalently, $T_1$~is included in~$T_2$ if $T_1$ can be obtained by applying
a sequence of \emph{node deletion operations} (defined as the inverse of
a node insertion operation) to~$T_2$.

A \emph{mapping} between two trees~$T_1$ and~$T_2$ is
a subset $M \subseteq V(T_1) \times V(T_2)$ such that for every
$(u_1,v_1), (u_2,v_2) \in M$, it holds that:
(i)~$u_1 = u_2$ if and only if $v_1 = v_2$; and
(ii)~$u_1$ is an ancestor of~$u_2$ if and only if $v_1$ is an ancestor
of~$v_2$.
Condition (i) states that each node appears at most once in $M$,
and condition (ii) states that ancestor-descendant relations must be preserved.
A mapping~$M$ between~$T_1$ and~$T_2$ such that $|M| = |V(T_1)|$ and
$u$ and~$v$ have the same node label for every $(u,v) \in M$ is called
an \emph{inclusion mapping} (see Fig.~\ref{fig:inclusion} for an example).
It is known that $T_1$~is included in~$T_2$ if and only if there exists
an inclusion mapping between~$T_1$ and~$T_2$~\cite{tai1979tree}.
We write $T_1(u) \subset T_2(v)$ if $T_1(u)$ is included in $T_2(v)$ under
the additional condition that there exists an inclusion mapping that maps~$u$
to~$v$.
For any two trees~$T_1$ and~$T_2$, $T_1 \sim T_2$ means that $T_1$ is
isomorphic to~$T_2$, in the sense that node labels have to be preserved.

In the \emph{tree inclusion problem}, the input is two trees~$P$ and~$T$,
also referred to as the ``pattern tree'' and the ``target tree'', and
the objective is to locate every minimal subtree of~$T$ that includes~$P$,
where $T(v)$ is called a \emph{minimal subtree} if it minimally includes $P$,
the definition of which is given below.
For any instance of the tree inclusion problem, we define $m = |V(P)|$ and
$n=|V(T)|$, and let $d$ denote the degree of~$P$, i.e., the maximum number of
children of any node in~$P$.
We assume w.l.o.g. (without loss of generality) that $m \leq n$ because
otherwise $P$ cannot be included in $T$.
The following concept plays a key role in our algorithm (see
Fig.~\ref{fig:inclusion} for an illustration).

\begin{definition}
For any instance of the tree inclusion problem and any $u \in V(P)$ and
$v \in V(T)$, $T(v)$~is said to \emph{minimally include~$P(u)$} (written as
$P(u) \prec T(v)$) if $P(u) \subset T(v)$ holds and there is no $v' \in Des(v)$
such that $P(u) \subset T(v')$.
\end{definition}

We may simply use $P$ and $T$ in place of $P(u)$ and $T(v)$
if $u$ and $v$ are the roots of $P$ and $T$, respectively.
Locating every minimal subtree is reasonable because 
$P(u) \subset T(v')$ holds for all ancestors $v'$ of $v$ if
$P(u) \prec T(v)$ holds.

\begin{figure}[t!]
\centering
\includegraphics[width=10cm]{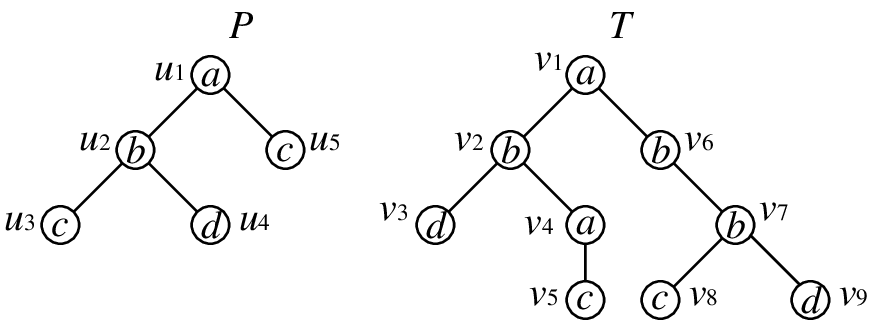}
\caption{An example of unordered tree inclusion.
Here, $P \subset T$ holds by an inclusion mapping 
$M = \{(u_1,v_1),(u_2,v_2),
(u_3,v_5),
(u_4,v_3),(u_5,v_8)\}$.
$P(u_2) \subset T(v_2)$,
$P(u_2) \subset T(v_6)$, and~$P(u_2) \subset T(v_7)$ hold as well.
Furthermore, $P \prec T$, $P(u_2) \prec T(v_2)$, and $P(u_2) \prec T(v_7)$
hold, but $P(u_2) \prec T(v_6)$~does not hold.}
\label{fig:inclusion}
\end{figure}

\begin{proposition}
Given any instance of the tree inclusion problem and any $u \in V(P)$ and
$v \in V(T)$ with $Chd(u) = \{u_1,\ldots,u_d\}$, it holds that
$P(u) \subset T(v)$ if and only if the following conditions are satisfied:
\begin{itemize}
\item [(1)] $\ell(u) = \ell(v)$;
\item [(2)] $v$ has a set of descendants $D(v) = \{v_1,\ldots,v_d\}$ such that
$v_i \notin Des(v_j)$ for every $i \neq j$; and
\item [(3)] there exists a bijection $\phi$ from $Chd(u)$ to $D(v)$ such that
$P(u_i) \prec T(v_i)$ holds for every $i \in \{1,2,\ldots,d\}$.
\end{itemize}
\label{prop:central}
\end{proposition}
\begin{proof}
Suppose that Conditions (1)-(3) are satisfied.
Condition (3) implies that there exists an injection mapping $\phi'$ between 
the forest induced by $u_1,\ldots,u_d$ and their descendants
and the forest induced by $v_1,\ldots,v_d$ and their descendants
such that $\phi'(u_i)=v_i$.
Let $\phi'' = \phi' \cup \{(u,v)\}$.
Since $u_1,\ldots,u_d$ are the children of $u$ and
$v_1,\ldots,v_d$ are descendants of $v$,
$\phi''$ is an inclusion mapping and thus $P(u) \subset T(v)$ holds.

Conversely, suppose that $P(u) \subset T(v)$ holds, which means
that there exists an inclusion mapping $\phi$ from $P(u)$ to $T(v)$
with $\phi(u)=v$.
Let $w_i=\phi(u_i)$ for $i=1,\ldots,d$.
Then, $w_i \notin Des(w_j)$ holds for every $i \neq j$
because $\phi$ is an inclusion mapping.
Furthermore,
for each $w_i$, there must exist $v_i \in \{w_i\} \cup Des(w_i)$
such that $P(u_i) \prec T(v_i)$ holds with an inclusion mapping $\phi_i$
from $P(u_i)$ to $T(v_i)$ satisfying $\phi_i(u_i)=v_i$.
Note that $v_i \notin Des(v_j)$ holds for every $i \neq j$ because
$w_i \notin Des(w_j)$ holds for every $i \neq j$.
Let $\phi' = \{(u,v)\} \cup \phi_1 \cup \ldots \phi_d$.
Condition (1) is satisfied because $(u,v) \in \phi'$.
Here we let $D(v)=\{v_1,\ldots,v_d\}$.
Then, Condition (2) is satisfied as stated above.
Condition (3) is also satisfied because
$\phi'(u_i) = v_i$ holds for all $i=1,\ldots,d$.
\end{proof}

Proposition~\ref{prop:central} essentially states that the children of~$u$
must be mapped to descendants of~$v$ that do not have ancestor-descendant
relationships.
Since $P$ is included in~$T$ if and only if there exists a $v \in V(T)$ with
$P \prec T(v)$, we need to determine if $P(u) \prec T(v)$, assuming that
whether $P(u_j) \prec T(v_i)$ holds is known for all $(u_j,v_i)$  with
$u_j \in Des(u) \cup \{ u \}$, $v_i \in Des(v) \cup \{ v \}$, and
$(u_j,v_i) \neq (u,v)$.
This assumption is satisfied if we apply a dynamic programming procedure
to determine if $P \prec T(v)$, 
using an $O(mn)$ size table and
following any partial ordering on $(u,v)$s in $V(P) \times V(T)$ such that
$(u,v)$ precedes $(u',v')$ if and only if
$u' \in Des(u) \cup \{u\}$,
$v' \in Des(v) \cup \{v\}$, and
$(u',v') \neq (u,v)$.

\begin{proposition}
Suppose that $P(u) \prec T(v)$ can be determined in $O(f(d,m,n))$ time,
assuming that whether $P(u_j) \prec T(v_i)$ holds is known for
all pairs $(u_j,v_i)$ such that
$(u_j,v_i) \in V(P(u)) \times V(T(v)) \setminus \{(u,v)\}$.
Then the unordered tree inclusion problem can be solved in
$O(f(d,m,n) \cdot mn)$ time by using a bottom-up dynamic programming procedure.
\label{prop:bottom-up}
\end{proposition}

\section{An $O(d \cdot 2^d \cdot mn^2)$-time algorithm}
\label{sec:improved}

The core of Kilpel\"{a}inen and Mannila's
algorithm~\cite{kilpelainen1995ordered} for unordered tree inclusion is
the computation of a set~$S(v)$ for each node $v \in V(T)$, also called
the \emph{match system for target node~$v$}.
In their paper,
$S(v)$~was originally defined as a set of subsets of nodes from~$P$,
where each such subset
consists of the root nodes in a subforest of~$P$ that is included in~$T(v)$.
However, 
$S(v)$ was restricted to subsets of $Chd(u)$ for a single node $u$ in
$P$
when the bounded outdegree case was considered.
We employ this restricted definition in this paper and
define~$S(v)$ for any fixed $u \in V(P)$ by:
\begin{eqnarray*}
S(v) & = & \{ U \subseteq Chd(u) \,|\, P(U) \subset T(v) \},
\end{eqnarray*}
where $P(U)$ is the forest induced by the nodes in~$U$ and their descendants
and $P(U) \subset T(v)$ means that every tree
from~$P(U)$ is included in~$T(v)$
without overlap
(i.e., $T(v)$ can be obtained from~$P(U)$ by node insertion operations).
For details, see~\cite{kilpelainen1995ordered}.

Kilpel\"{a}inen and Mannila's algorithm~\cite{kilpelainen1995ordered} computes
the $S(v)$-sets in a bottom-up order.
It fixes an arbitrary left-to-right ordering of the nodes of~$T$ (the ordering
will not affect the correctness).
Precisely, the left-to-right ordering is determined as follows.
We assume that for each node having two or more children,
a left-to-right ordering is given to the children.
For any two nodes $v_i, v_j \in V(T)$ (resp., $v_i,v_j \in V(P)$)
that do not have any ancestor-descendant relationship,
let $v$ be the lowest common ancestor, which is uniquely determined.
For any descendant $v_k$ of $v$,
let $v_k'$ be the child of $v$ such that $v_k$ is a descendant of $v_k'$
or $v_k=v_k'$.
Then, $v_i$ is left (resp., right) of $v_j$ if and only if
$v_i'$ is a left (resp., right) sibling of $v_j'$.
Note that left-right relationships are defined for nodes
only if they do not have any ancestor-descendant relationship.
Below, we denote ``$v_i$~is left of~$v_j$'' by $v_i \triangleleft v_j$.
To compute~$S(v)$, their algorithm performs the following operation from left
to right for the children $v_1,\ldots,v_l$ of~$v$:
\begin{eqnarray*}
S & := & \{A \cup B \,|\, A \in S,\, B \in S(v_i) \},
\end{eqnarray*}
starting with $S = \emptyset$, and then $S(v)$ is assigned the resulting~$S$.
Clearly, the size of~$S(v)$ is no greater than~$2^d$.
However, this way of updating~$S$ causes an $\Omega(2^{2d})$-factor in
the running time because it examines $\Omega(2^d) \times \Omega(2^d)$
set pairs.
To avoid this bottleneck, we need a new approach for computing~$S(v)$,
explained next.
We shall focus on how to determine if $P(u) \prec T(v)$ holds for
a fixed $(u,v)$ because this part is crucial for reducing the time complexity.

Assume w.l.o.g.
that $u$ has $d$~children and
write $Chd(u) = \{u_1,\ldots,u_d\}$.
To simplify the presentation, we will assume until the end of this section that
$P(u_i) \sim P(u_j)$ does not hold for any $u_i, u_j \in Chd(u)$ with
$u_i \neq u_j$.
For any $v_i \in V(T(v))$, define $M(v_i)$ by:
\begin{eqnarray*}
M(v_i) & = &  \{ u_j \in Chd(u) \,|\, P(u_j) \prec T(v_i) \}.
\end{eqnarray*}
For example, $M(v_0) = \emptyset$, $M(v_2)=\{u_C\}$,
and $M(v_3)=\{u_D,u_E\}$ in Fig.~\ref{fig:key-idea}.
Note that $M(v_i)$ is known for all descendants $v_i$ of $v$ before testing
$P(u) \prec T(v)$ and does not change during the course of this testing.
For any $v_i \in V(T(v))$, $LF(v,v_i)$~denotes the set of nodes in~$V(T(v))$
each of which is left of~$v_i$
(see Fig.~\ref{fig:key-idea} for an example).
Next, define $S(v,v_i)$ by:
\begin{eqnarray*}
S(v,v_i) & = & \{ U \subseteq Chd(u) \,|\, P(U) \subset T(LF(v,v_i)) \} \cup \\
& & \{ U \subseteq Chd(u) \,|\, (U=U' \cup \{u_j\}) \land (P(U') \subset T(LF(v,v_i))) \land\\
& & ~~~~~~~~~~~~~~~~~~~~~~~~~~~~~~~~~~~~~~~~~~~~~(u_j \in M(v_i)) \}
\end{eqnarray*}
where $T(LF(v,v_i))$ is the forest induced by the nodes in $LF(v,v_i)$ and
their descendants.
Note that $P(\emptyset) \subset T(...)$ always holds.
Note also that each element of $S(v,v_i)$ is a subset of the children of $u$
that are included in the forest induced by the nodes left to $v_i$ and
in $V(T(v_i))$ under the condition that there exists at most one child $u_j$
such that $P(u_j)$ is included in $T(v_i)$ in the corresponding
inclusion mapping.
The motivation for introducing $S(v,v_i)$ is that Lemma~\ref{lemma:S_v}
below will allow us to recover $S(v)$ from a collection of $S(v,v_i)$-sets,
and the $S(v,v_i)$-sets can be computed efficiently with dynamic programming.
We explain $S(v,v_i)$ using an example based on Fig.~\ref{fig:key-idea}.
Suppose we have the relations
$P(u_A) \prec T(v_1)$,
$P(u_B) \prec T(v_1)$,
$P(u_C) \prec T(v_2)$,
$P(u_D) \prec T(v_3)$,
$P(u_E) \prec T(v_3)$,
$P(u_D) \prec T(v_4)$,
and $P(u_F) \prec T(v_4)$.
Then, the following holds:
\begin{tabbing}
$S(v,v_0) = \{~\emptyset ~\}$, \\
$S(v,v_1) = \{~\emptyset,~\{u_A\},~\{u_B\} ~\}$, \\
$S(v,v_2) = \{~\emptyset,~\{u_C\}~\}$, \\
$S(v,v_3) = \{~\emptyset,~\{u_D\},~\{u_E\} ~\}$, \\
$S(v,v_4) = \{~\emptyset,~\{u_D\},~\{u_E\},~\{v_F\},~\{u_D,u_E\},~\{u_D,u_F\},
~\{u_E,u_F\} ~\}$.
\end{tabbing}

\begin{figure}[t!]
\centering
\includegraphics[width=13cm]{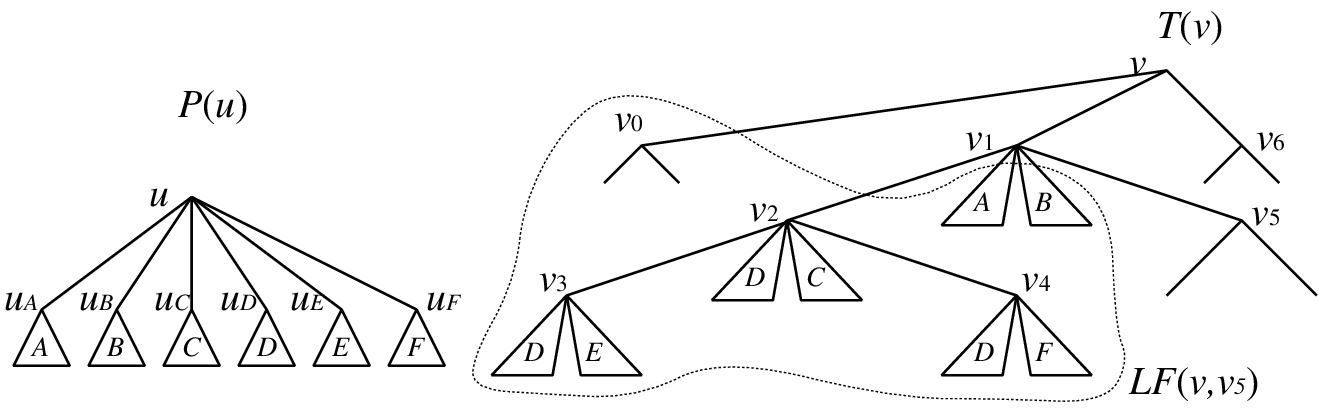}
\caption{An example.
A triangle~$X$ attached to~$v_i$ means that $P(u_X) \subset T(v_i)$.
Note that the triangle~$D$ appears at~$v_2$, $v_3$, and~$v_4$.
However, $P(u_D) \prec T(v_2)$ does not hold since it does not satisfy
the minimality condition.
Therefore, $u_D$~may be matched to~$v_3$ or~$v_4$, but $u_D$~will never be
matched to~$v_2$ in \alga.}
\label{fig:key-idea}
\end{figure}

\bigskip

Next, we present a dynamic programming-based algorithm named \alga, for
determining if $P(u) \prec T(v)$.
To compute all the $S(v,v_i)$-sets, we construct
a DAG (directed acyclic graph) $G(V,E)$ from~$T(v)$, as illustrated in
Fig.~\ref{fig:dag}.
Here, $V$~is defined by $V = V(T(v)) - \{ v \}$, and $E$ is defined by
$E = \{ (v_i,v_j) \,|\, v_i \triangleleft v_j \}$.
We define $Pred(v_i)$ by
$Pred(v_i) = \{v_j \,|\, (v_j,v_i) \in E \}$,
meaning the set of the ``predecessors'' of~$v_i$,
and also being equivalent to $LF(v,v_i)$.
\alga~traverses $G(V,E)$ so that node~$v_i$ is visited only after all of
its predecessors have been visited, at which point it runs
the procedure \proca$(v,v_i)$~below
to compute and store~$S(v,v_i)$ for this~$v_i$.
Recall that $M(v_i) = \{ u_j \in Chd(u) \,|\, P(u_j) \prec T(v_i) \}$.

\begin{figure}[t!]
\centering
\includegraphics[width=7cm]{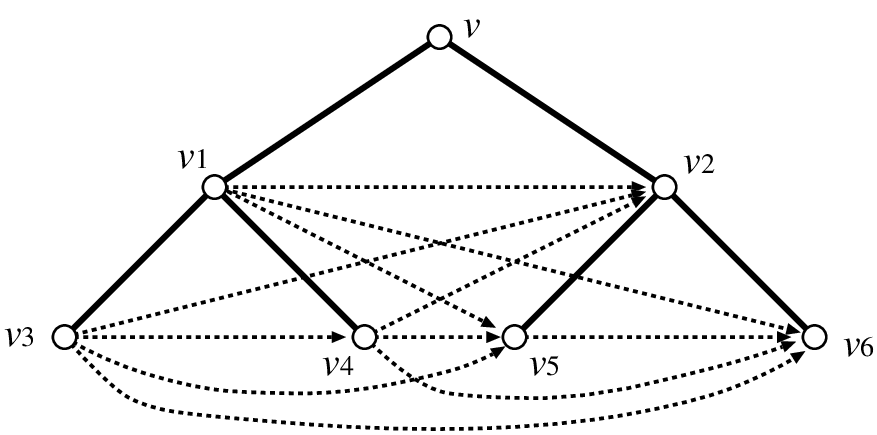}
\caption{Example of the DAG $G(V,E)$ constructed from~$T(v)$, where
$v \notin V$, $E$~is shown by dotted arrows, and $T(v)$~is shown by
bold lines.}
\label{fig:dag}
\end{figure}

\bigskip
\noindent
Procedure \proca$(v,v_i)$:

\begin{itemize}
\item [(1)\,\,\,]
If $Pred(v_i) = \emptyset$ then
$S(v,v_i) := \{ \emptyset \} \cup \{ \{u_h\} \,|\, u_h \in M(v_i) \}$
\item [(2)\,\,\,]
Else:
\item [(2a)]
~ ~ $S_0(v_i) := \bigcup_{v_j \in Pred(v_i)} S(v,v_{j})$
\item [(2b)]
~ ~ $S(v,v_i) := S_0(v_i) \cup
\{ S \cup \{ u_h \} \,|\, u_h \in M(v_i), S \in S_0(v_i) \}$
\end{itemize}

Finally, after $G(V,E)$ has been completely traversed, \alga~assigns
$S(v) := \bigcup_{v_i \in Des(v)} S(v,v_i)$.
Then $P(u)$ is included in $T(v)$ with $u$ corresponding to~$v$ if and only if
$u$~and $v$ have the same label and $Chd(u) \in S(v)$.
Note that $S(v)=\emptyset$ holds for each $v$ if $Chd(u)=\emptyset$.

\begin{lemma}
Procedure \proca$(v,v_i)$ correctly computes $S(v,v_i)$s, and
$S(v) = \cup_{v_i \in Des(v)} S(v,v_i)$.
\label{lemma:S_v}
\end{lemma}

\begin{proof}
First we show that \proca$(v,v_i)$ correctly computes $S(v,v_i)$s.
It is seen from Proposition~\ref{prop:central} that
$U \in S(v,v_i)$ holds for $U=\{u_{i_1},\ldots,u_{i_k}\} \subseteq Chd(u)$
($U = \emptyset$ if $k=0$)
if and only if there exists a sequence of nodes
$(v_{j_1},v_{j_2},\ldots,v_{j_k})$
such that
$P(u_{i_p}) \prec T(v_{j_p})$ holds for all $p=1,\ldots,k$ and 
$v_{j_1} \triangleleft v_{j_2} \triangleleft \cdots \triangleleft v_{j_k}$
(by appropriately renumbering indices of $u_{i_1},\ldots,u_{i_k}$),
where $v_{j_k} = v_i$ or $v_{j_k} \triangleleft v_i$.
On the other hand, it is seen from \proca$(v,v_i)$ that
this procedure examines all possible sequences such that
$v_{j_1} \triangleleft v_{j_2} \triangleleft \cdots \triangleleft v_{j_{k'}}$
with $v_{j_{k'}} = v_i$ or $v_{j_{k'}} \triangleleft v_i$,
and adds at most one $u_h \in M(v_{j_p})$ to each set in $S_0(v_{j_p})$.
It is also seen that the procedure and the above discussion that
$S_0(v_{j_p})$ consists of $U$s such that
$U \subseteq Chd(u)$ and $P(U) \subset T(LF(v,v_{j_p}))$.
Therefore, we can see from the definition of $M(\ldots)$ that
\proca$(v,v_i)$ correctly computes $S(v,v_i)$s.

Then we show the second statement of the lemma.
Let $U \in S(v)$ and $d_U=|U|$.
Let $\phi$ be an injection from $U$ to $Des(v)$ giving
an inclusion mapping for $P(U) \subset T(v)$,
which is the one guaranteed by Proposition~\ref{prop:central}.
Let $\{v_1',\ldots,v_{d_U}'\} = \{ \phi(u_j) | u_j \in U \}$,
where $v_1' \triangleleft v_2' \triangleleft \cdots \triangleleft v_{d_U}'$.
Then,
$v_i' \in LF(v,v_{i+1}')$ and $v_i' \in LF(v,v_{d_U}')$ hold
for all $i=1,\ldots,d_U-1$.
Furthermore, $P(u_j) \prec T(v_i')$ holds for $v_i' = \phi(u_j)$.
Therefore, $U \in S(v,v_{d_U}')$.

It is straightforward to see that $S(v,v_i)$ does not contain any element
not in~$S(v)$.
\end{proof}

The overall procedure of \alga is given by the pseudocode of
{\bf Algorithm}~\ref{alg:treeincl1}.
In this procedure, we traverse nodes in both $P$ and $T$ from left to right
n the postorder (i.e., leaves to the root).
We maintain $Min(u)$ for $u \in V(P)$ (resp., $Min(v)$ for $v \in V(T)$)
that consists of the currently available nodes $v'$ (resp., $u'$) such that
$P(u) \prec T(v')$ (resp., $P(u') \prec T(v)$).

\begin{algorithm}
\caption{\alga$(P,T)$}
\label{alg:treeincl1}
\begin{algorithmic}
\FOR{all $v \in V(P) \cup V(T)$}
\STATE{$Min(v):=\emptyset$}
\ENDFOR
\FOR{all $u \in V(P)$ in postorder}
\FOR{all $v \in V(T)$ in postorder}
\STATE{$M(v):=\{ u_j \in Chd(u)| u_j \in Min(v)\}$}
\FOR{all $v_i \in Des(v)$ in postorder}
\STATE{{\bf ComputeSet$(v.v_i)$}}
\ENDFOR
\STATE{$S(v) := \cup_{v_i \in Des(v)} S(v,v_i)$}
\IF{$Chd(u) \in S(v)$ \AND $u \notin Min(v_i)$ for all $v_i \in Des(v)$ \AND $\ell(u)=\ell(v)$}
\STATE{$Min(v):=Min(v)\cup\{u\}$}
\STATE{$Min(u):=Min(u)\cup\{v\}$}
\ENDIF
\ENDFOR
\ENDFOR
\RETURN $Min(r(P))$
\end{algorithmic}
\end{algorithm}


\begin{lemma}
\alga~outputs the set of all nodes $v$ such that
$P \prec T(v)$ in $O(d \cdot 2^d \cdot m n^3)$ time using
$O(d \cdot 2^d \cdot n + mn)$ space.
\label{lemma:alga}
\end{lemma}

\begin{proof}
Since the correctness follows from Lemma~\ref{lemma:S_v},
we analyze the time complexity.

The sizes of the $S(v)$, $S(v,v_{i_j})$s, and $S_0(v_i)$s are
$O(d \cdot 2^d)$, where we can use a simple bit vector of size $O(d)$ to
represent each subset of $U$.
The computation of each of these sets takes $O(d \cdot 2^d \cdot n)$ time.
Since the number of $S(v,v_{i_j})$s and $S_0(v_i)$s per
$(u,v) \in V(P) \times V(T)$ are $O(n)$,
the total computation time for $S(v,v_{i_j})$s per $(u,v)$ 
is $O(d \cdot 2^d \cdot n^2)$.
Hence, the total computation time for computing $S(v)$s for
all $(u,v)$s is $O(d \cdot 2^d \cdot m n^3)$.

Since the size of each $S(v,v_i)$ is $O(d \cdot 2^d)$ and we need to maintain
$S(v,v_i)$ for $v_i \in Des(v)$ per $(u,v)$,
$O(d \cdot 2^d \cdot n)$ space is enough to maintain $S(v,v_i)$s.
Note that we can re-use the same space for different $(u.v)$s.

The time needed for other operations can be analyzed as follows.
We can use simple bit vectors to maintain $Min(u)$s and $Min(v)$s,
which need $O(mn)$ space in total and $O(1)$ time per addition of an element
or checking of the membership.
Therefore, the total computation time required to maintain $Min(u)$s and $Min(v)$s is $O(mn)$.
Furthermore, $M(v)$ can be computed in $O(d)$ time per $(u,v)$ and thus
the total time to compute $M(v)$s is $O(dmn)$, and
``$u \notin Min(v_i)$ for all $v_i \in Des(v)$'' can be checked in
$O(|Des(v)|) \leq O(n)$ time per $(u,v)$ and thus
the total computation time needed for this checking is $O(mn^2)$.

Therefore, the time and space complexities
of \alga~are $O(d \cdot 2^d \cdot m n^3)$ and $O(d \cdot 2^d \cdot n + mn)$,
respectively.
\end{proof}

\textbf{Remark:}
If there exist $u_i, u_j \in Chd(u), u_i \neq u_j$ such that
$P(u_i) \sim P(u_j)$, we treat each element in $S(v)$, $S(v,v_{i_j})$s,
and $S_0(v_i)$s as a multiset
where any $u_i$ and $u_j$ such that $P(u_i) \sim P(u_j)$
are identified and
the multiplicity of $u_i$ is bounded by the number of $P(u_j)$s isomorphic
to $P(u_i)$.
Then, since $|Chd(u)| \leq d$ for all $u$ in $P$,
the size of each multiset is at most $d$ and
the number of different multisets is not greater than $2^d$.
Therefore, the same time complexity result holds.
(The same arguments can be applied to the following sections.)
Note that by treating~$u_i$ and~$u_j$ separately, we do not need to modify
the algorithm.

\begin{figure}[t!]
\centering
\includegraphics[width=14cm]{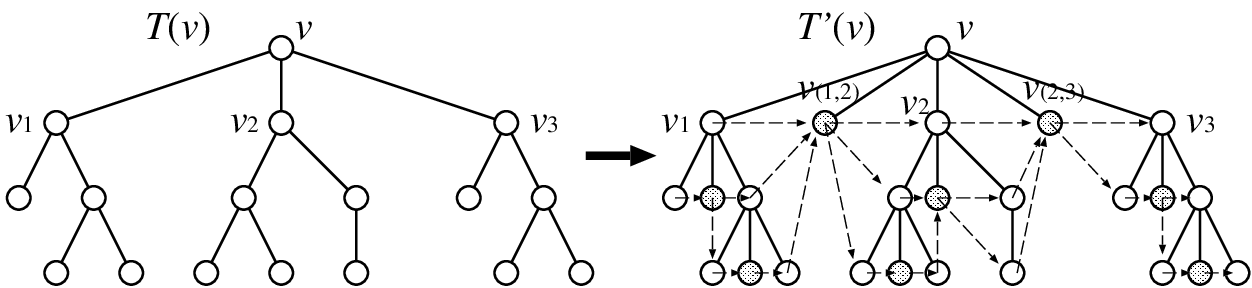}
\caption{Example of $T'(v)$ and $G'(V',E')$.
$E'$~is shown by dashed arrows.}
\label{fig:new-tree}
\end{figure}

Next, we discuss how to improve the efficiency of \alga.
Actually, to compute $S_0(v_i)$, it is not necessary to consider all of
the $v_{i_j}$s that are left of~$v_i$.
Instead, we can construct a tree $T'(v)$ from a given~$T(v)$ according to
the following rule (see Fig.~\ref{fig:new-tree} for an illustration):
\begin{itemize}
\item
For each pair of consecutive siblings $(v_i,v_j)$ in $T(v)$, add
a new sibling (leaf) $v_{(i,j)}$ between $v_i$ and~$v_j$.
\end{itemize}
Newly added nodes are called \emph{virtual nodes}.
All virtual nodes have the same label that does not appear in $P$,
to ensure that no $u \in V(P)$ is in $M(v_{(i,j)})$.
%
We then construct a DAG $G'(V',E')$ on $V'=V(T'(v))$ where
$(v_i,v_j) \in E'$ if and only if one of the following holds:
\begin{itemize}
\item $v_j$ is a virtual node, and $v_i$ is in the rightmost path
of~$T'(v_{j_1})$, where $v_j = v_{(j_1,j_2)}$; or
\item $v_i$ is a virtual node, and $v_j$ is in the leftmost path
of~$T'(v_{i_2})$, where $v_i = v_{(i_1,i_2)}$.
\end{itemize}
By replacing $G(V,E)$ by~$G'(V',E')$ in \alga~(and keeping all other steps
intact),
we obtain what we call \algb.
Note that in \algb, $v_{(i,j)}$s are treated in the same ways as for $v_i$s
and thus we need not introduce the definitions for such terms
as $S(v.v_{(i,j)})$ and $LF(v,v_{(i,j)})$
nor change the definition of $S(v.v_i)$.

\begin{figure}[t!]
\centering
\includegraphics[width=10cm]{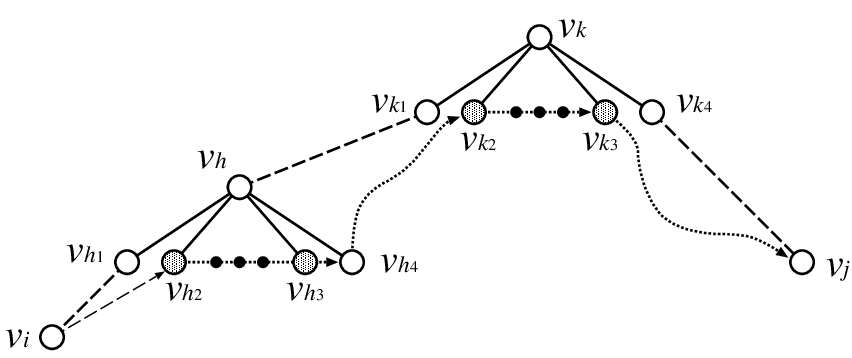}
\caption{Illustration of a path from $v_i$ to $v_j$ in the proof
of Lemma~\ref{lemma:algb}.}
\label{fig:path}
\end{figure}

\begin{lemma}
\algb~computes $S(v,v_i)$s for all $v_i \in Des(v)$ in
$O(d \cdot 2^d \cdot n)$ time per $(u,v) \in V(P) \times V(T)$.
\label{lemma:algb}
\end{lemma}
\begin{proof}
First we prove that there exists a path in $G'(V',E')$
from $v_i \in V$ to $v_j \in V$ if and only if
$v_i \triangleleft v_j$ (see also Fig.~\ref{fig:path}).
It can be seen
from the definition of the left-right relationship
that if $(v_i,v_k) \in E'$ and $(v_k,v_j) \in E'$
where $v_i,v_j \in V$ and $v_k$ is a virtual node,
then $v_i \triangleleft v_j$.
Since virtual nodes and non-virtual nodes appear alternatively
in every path in $G'(V',E')$,
the ``only if'' part holds.
Suppose that $v_i \triangleleft v_j$ holds for $v_i,v_j \in V$.
Let $v_k$ be the lowest common ancestor of $v_i$ and $v_j$.
We assume w.l.o.g. that $v_i$ or $v_j$ is not a child of $v_k$
because the other cases
can be proved in the same way.
Let $v_{k_1},v_{k_2},v_{k_3},v_{k_4}$ be children of $v_k$
such that $v_{k_1} \in Anc(v_i)$,
$(v_{k_1},v_{k_2}) \in E'$,
$(v_{k_3},v_{k_4}) \in E'$,
and 
$v_{k_4} \in Anc(v_j)$,
where $v_{k_2}$ and $v_{k_3}$ are virtual nodes and can be the same node.
We show that there exists a path in $G'(V',E')$ from $v_i$ to $v_{k_2}$.
Let $v_h$ be the lowest ancestor of $v_i$ that has
children $v_{h_1},v_{h_2},v_{h_3},v_{h_4}$
such that $v_{h_1} \in Ans(v_i)$,
$(v_{h_1},v_{h_2}) \in E'$,
$(v_{h_3},v_{h_4}) \in E'$,
and
$v_{h_4}$ ($\neq v_{h_1}$) is the rightmost child of $v_h$,
where $v_{h_2}$ and $v_{h_3}$ are virtual nodes and can be the same node.
Then, $(v_i,v_{h_2}) \in E'$ holds from the construction of $G'(V',E')$
and thus there exists a path from $v_i$ to $v_{h_4}$.
We can repeat this procedure by regarding $v_{h_4}$ as $v_i$, and so on,
from which it follows that
there exists a path in $G'(V',E')$ from $v_i$ to $v_{k_2}$.
It is also seen from the symmetry on the left-right relationship
that there exists a path in $G'(V',E')$ from $v_{k_3}$ to $v_j$.
Furthermore, there clearly exists a path in $G'(V',E')$ from $v_{k_2}$ to $v_{k_3}$, which completes
the proof of the ``if'' part.

Moreover, from the above discussion, it can be seen that
\algb~examines the same set of sequences 
$v_{j_1} \triangleleft v_{j_2} \triangleleft \cdots \triangleleft v_{j_{k'}}$
as \alga~examines when ignoring virtual nodes.
Furthermore, 
addition of an element is not performed at any virtual node, and
no element is deleted at any virtual or non-virtual node $v$
in constructing $S(v)$.
Therefore, \algb~correctly computes $S(v,v_i)$s.

Next
we analyze the time complexity.
We can see that $|E'| = O(n)$ since:
\begin{itemize}
\item $|V(T'(v))| = O(n)$;
\item each non-virtual node in $G'(V',E')$ has at most one incoming edge and
at most one outgoing edge; and
\item each
new edge
connects a non-virtual node and virtual node.
\end{itemize}
Therefore, the total number of set operations is $O(d \cdot 2^d \cdot n)$, and
the lemma follows.
\end{proof}


From Lemmas \ref{lemma:alga} and \ref{lemma:algb},
we have the following main theorem.

\begin{theorem}
Unordered tree inclusion can be solved in $O(d \cdot 2^d \cdot mn^2)$ time
and $O(d \cdot 2^d \cdot n + mn)$ space.
\label{thm:main}
\end{theorem}

If we use the height~$h(T)$ of a tree~$T$ as an additional parameter,
we can express the time complexity as $O(d \cdot 2^d \cdot h(T) \cdot mn)$
because
the time complexity is represented in this case as
$O(m \sum_{v \in V(T)} d \cdot 2^d \cdot |T(v)|)$
and $\sum_{v \in V(T)} |T(v)| \leq (h(T)+1)n$
hold.
This bound is better than
the one by Kilpel\"{a}inen and Mannila~\cite{kilpelainen1995ordered}
when $d$ is
large (to be precise, when $d > c \log(h(T))$ for some constant~$c$).

\section{NP-hardness of the case of pattern trees with unique leaf labels}
\label{sec:hardness}

For any node-labeled tree~$T$, let $L(T)$ be the set of all leaf labels in~$T$.
For any $c \in L(T)$, let $occ(T,c)$ be the number of times that~$c$
occurs in~$T$, and define $occ(T) = \max_{c \in L(T)} occ(T,c)$.

The decision version of the tree inclusion problem is the problem of
determining whether~$T$ can be obtained from~$P$ by applying a sequence of
node insertion operations.
Kilpel\"{a}inen and Mannila~\cite{kilpelainen1995ordered} proved that
the decision version of unordered tree inclusion is NP-complete by a reduction
from Satisfiability.
In their reduction, the clauses in a given instance of Satisfiability are used
to label the non-root nodes in the constructed trees~$P$ and~$T$;
in particular, for every clause~$C$, each literal in~$C$ introduces one node
in~$T$ whose node label represents~$C$.
(See the proofs of Lemma~7.2 and Theorem~7.3
in~\cite{kilpelainen1995ordered} by Kilpel\"{a}inen and Mannila
for details.)
By using 3-SAT instead of Satisfiability in their reduction, every clause will
determine the label of at most three nodes in~$T$, so we immediately have:

\begin{corollary}
\label{corollary:KM_hardness}
The decision version of the unordered tree inclusion problem is NP-complete
even if restricted to instances where
$h(T) = 2$, $h(P) = 1$, $occ(T) = 3$, and $occ(P) = 1$.
\end{corollary}

In Kilpel\"{a}inen and Mannila's reduction, the labels assigned to
the internal nodes of~$T$ are significant.
Here, we consider the computational complexity of the special case of
the problem where all internal nodes in~$P$ and~$T$ have the same label, or
equivalently, where only the leaves are labeled.
The next theorem is the main result of this section.

\begin{theorem}
\label{thm:unique_leaves_NP-complete}
The decision version of the unordered tree inclusion problem is NP-complete
even if restricted to instances where
$h(T) = 2$, $h(P) = 2$, $occ(T) = 3$, $occ(P) = 1$, and all internal nodes
have the same label.
\end{theorem}
\begin{proof}
Membership in NP was shown in the proof of
Theorem~7.3 by Kilpel\"{a}inen and Mannila~\cite{kilpelainen1995ordered}.
Next, to prove the NP-completeness, we present a reduction from
\textsc{Exact Cover by 3-Sets (X3C)}, which is known to
be NP-complete~\cite{book:GarJoh79}.
\textsc{X3C} is defined as follows.

\bigskip

\noindent
\fbox{
\parbox{\boxwidth}{
\smallskip
\textsc{Exact Cover by 3-Sets (X3C)}:

\medskip

Given a set $U = \{u_1, u_2, \dots, u_n\}$ and
a collection $\SSS = \{S_1, S_2, \dots, S_m\}$ of subsets of~$U$ where
$|S_i| = 3$ for every $S_i \in \SSS$ and every $u_i \in U$ belongs to at most
three subsets in~$\SSS$,
does $(U,\SSS)$ admit an exact cover, i.e., is there an $\SSS' \subseteq \SSS$
such that $|\SSS'| = n/3$ and $\bigcup_{S_i \in \SSS'} S_i = U$?
\smallskip
}
}

\medskip

We assume w.l.o.g. that in any given instance of \textsc{X3C}, $n/3$~is
an integer and each $u_i \in U$ belongs to at least one subset in~$\SSS$.

Given an instance $(U,\SSS)$ of \textsc{X3C}, construct
two node-labeled, unordered trees~$T$ and~$P$ as described next.
(Refer to Fig.~\ref{fig:unique_leaves_example} for an example of
the reduction.)
Let $W = \{s_i^j \,:\, 1 \leq i \leq m,\, 0 \leq j \leq n/3\}$ be a set of
elements different from~$U$ (i.e., $U \cap W = \emptyset$),
define $L = U \cup W$, and let $\alpha$ be
an element not in~$L$.
For any $L' \subseteq L$, let $t(L')$ denote the height-$1$ unordered tree
consisting of a root node labeled by~$\alpha$ whose children are bijectively
labeled by~$L'$.
Construct $T$ by creating a node~$r$ labeled by~$\alpha$ and attaching
the roots of the following trees as children of~$r$:
\begin{itemize}
\item[(i)] $t(\{s_i^0\} \cup S_i)$ for each $i \in \{1,2,\dots,m\}$
\item[(ii)] $t(\{s_i^{j-1},s_i^{j}\})$ for each $i \in \{1,2,\dots,m\}$,
$j \in \{1,2,\dots,n/3\}$
\item[(iii)] $t(\{s_1^j,s_2^j,\dots,s_m^j\})$ for each
$j \in \{1,2,\dots,n/3\}$
\end{itemize}
Construct $P$ by taking a copy of~$t(U)$ and then, for each $w \in W$,
attaching the root of~$t(\{w\})$ as a child of the root of~$P$.
Note that by construction, $L(T) = L(P) = L$, $h(T) = 2$, $h(P) = 2$,
$occ(T) = 3$, and $occ(P) = 1$ hold.

\begin{figure}[t]
\centering
\includegraphics[scale=0.46]{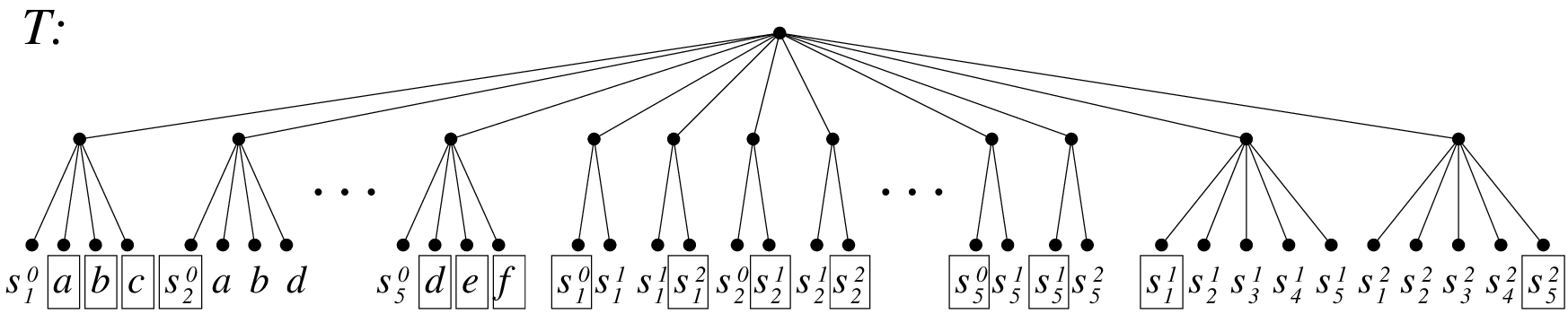}

\bigskip

\includegraphics[scale=0.46]{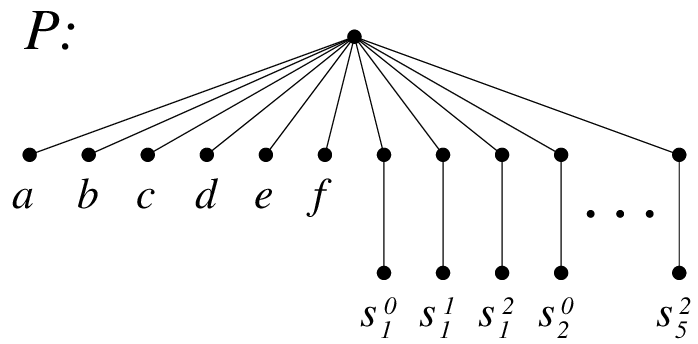}
\caption{Illustrating the proof of
Theorem~\ref{thm:unique_leaves_NP-complete}.
Suppose that $U = \{a,b,c,d,e,f\}$ and
$\SSS = \{\{a,b,c\}, \{a,b,d\}, \{b,c,e\}, \{c,d,e\},\{d,e,f\}\}$ with $|\SSS| = 5$
is a given instance of \textsc{X3C}.
Applying the reduction yields the shown trees~$T$ and~$P$.
Here, $P$~is included in~$T$ because all the leaves of~$P$ can be mapped to
leaves in~$T$ as indicated by the boxes, which gives
the exact cover $\{\{a,b,c\}, \{d,e,f\}\}$ for~$(U,\SSS)$.}
\label{fig:unique_leaves_example}
\end{figure}

\medskip

We will now show that $P$ is included in~$T$ if and only if $(U,\SSS)$ admits
an exact cover.

\medskip

\noindent
$(\leftarrow)$
First, suppose that $(U,\SSS)$ admits an exact cover
$\{S_{\sigma_1}, S_{\sigma_2}, \dots, S_{\sigma_{n/3}}\} \subseteq \SSS$.
Then $P$ is included in~$T$ because:
\begin{itemize}
\item[{\raise0.9pt\hbox{$\bullet$}}]
For each $S_i \in \SSS$ in the exact cover, the three leaves in~$P$ that
are labeled by~$S_i$ can be mapped to the $t(\{s_i^0\} \cup S_i)$-subtree
in~$T$.

\item[{\raise0.9pt\hbox{$\bullet$}}]
For each $S_i \in \SSS$ in the exact cover, the leaf in~$P$ labeled by~$s_i^j$
can be mapped
to the $t(\{s_i^{j},s_i^{j+1}\})$-subtree in~$T$ for $j \in \{0,1,\dots,k-1\}$,
to the $t(\{s_1^j,s_2^j,\dots,s_m^j\})$-subtree for $j = k$, and
to the $t(\{s_i^{j-1},s_i^{j}\})$-subtree for $j \in \{k+1,k+2,\dots,n/3\}$,
where $k$ is defined by $S_i = S_{\sigma_k}$.

\item[{\raise0.9pt\hbox{$\bullet$}}]
For each $S_i \in \SSS$ that is not in the exact cover, the leaf in~$P$ labeled
by~$s_i^0$ can be mapped to the $t(\{s_i^0\} \cup S_i)$-subtree in~$T$.

\item[{\raise0.9pt\hbox{$\bullet$}}]
For each $S_i \in \SSS$ that is not in the exact cover, the leaf in~$P$ labeled
by~$s_i^j$ can be mapped to the $t(\{s_i^{j-1},s_i^{j}\})$-subtree in~$T$ for
$j \in \{1,2,\dots,n/3\}$.
\end{itemize}

\medskip

\noindent
$(\rightarrow)$
Next, suppose that $P$ is included in~$T$.
By the definitions of~$T$ and~$P$, each subtree rooted at a child of
the root of~$T$ can
have at most one leaf with a label in~$W$ or at most three leaves with labels
in~$U$ mapped to it from~$P$.
Since $|W| = m \cdot (n/3 + 1)$ but there are only
$(m+1) \cdot n/3$ subtrees in~$T$ of the form $t(\{s_i^{j-1},s_i^{j}\})$
and $t(\{s_1^j,s_2^j,\dots,s_m^j\})$, at least $m - n/3$ subtrees of the form
$t(\{s_i^0\} \cup S_i)$ must have a leaf with a label from
$\{s_i^0 \,:\, 1 \leq i \leq m\}$ mapped to them.
This means that at most $n/3$ subtrees of the form $t(\{s_i^0\} \cup S_i)$
remain for the $n$~leaves in~$P$ labeled by~$U$ to be mapped to,
and hence, exactly $n/3$ such subtrees have to be used.
Denote these $n/3$ subtrees by
$t(\{s_{\sigma_1}^0\} \cup S_{\sigma_1})$,
$t(\{s_{\sigma_2}^0\} \cup S_{\sigma_2})$,
$\dots$,
$t(\{s_{\sigma_{n/3}}^0\} \cup S_{\sigma_{n/3}})$.
Then $\{S_{\sigma_1}, S_{\sigma_2}, \dots, S_{\sigma_{n/3}}\}$ is
an exact cover of~$(U,\SSS)$.
\end{proof}

\section{A polynomial-time algorithm for the case of $occ(P,T)=2$}
\label{sec:poly}

This section and the following ones consider the decision version of
unordered tree inclusion.
By repeatedly applying each procedure $O(n)$ times, we can solve
the locating problem version and thus the theorems hold as they are.

In this section, we require that each leaf of $P$ has a unique label
and that it appears at no more than $k$ leaves in $T$.
We denote this number $k$ by $occ(P,T)$ (see Fig.~\ref{fig:d2d3}).
Note that the case of $occ(P)=1$ and $occ(T)=k$ is included in
the case of $occ(P,T)=k$.
From the unique leaf label assumption, we have the following observation.

\begin{figure}[t!]
\centering
\includegraphics[width=12cm]{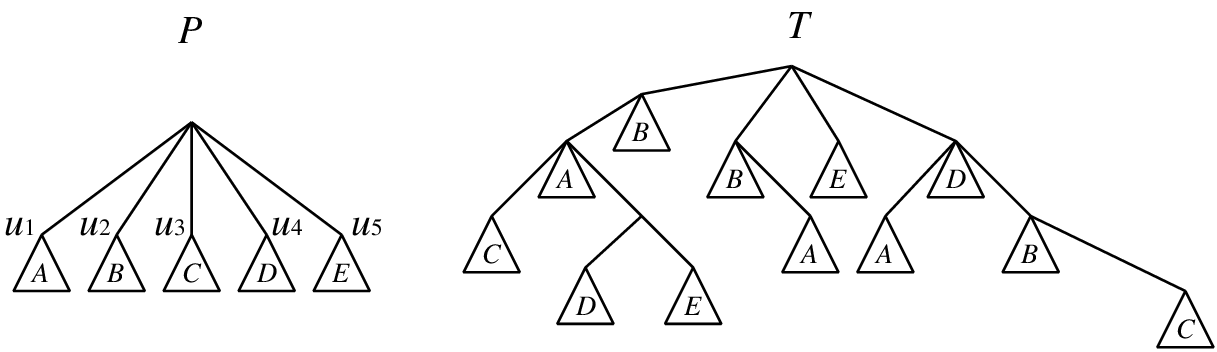}
\caption{For these trees, $Occ(u_1,M)=Occ(u_2,M)=3$,
$Occ(u_3,M)=Occ(u_4,M)=Occ(u_5,M)=2$,
$d_2=3$, $d_3=2$, and $occ(P,T)=3$,}
\label{fig:d2d3}
\end{figure}

\begin{proposition}
Suppose that $P(u)$ has a leaf labeled with $b$.
If $P(u) \subset T(v)$, then $v$ is an ancestor of a leaf (or leaf itself)
with label $b$.
\end{proposition}

We say that $v_j$ is a \emph{minimal node for $u_i$} if $P(u_i) \prec T(v_j)$
holds.
It follows from the proposition above that
the number of minimal nodes is at most $k$ for each $u_i$ if $occ(P,T)=k$.

The preliminary version of this paper \cite{akutsu2018unordered} showed
that the case $k=2$ can be solved in polynomial time by using a reduction
to 2-SAT.
Here, we give a more direct solution that effectively utilizes some techniques
from a classic polynomial-time algorithm for 2-SAT~\cite{aspvall1979}.
This algorithm will be extended for the case of $k=3$ in the next subsection.

From Proposition~\ref{prop:bottom-up},
it is enough to consider the decision of whether $P(u) \subset T(v)$ with
$u$ corresponding to $v$.
Let $Chd(u) = \{u_1,\ldots,u_d\}$.
We present a simple algorithm to decide whether or not $P(u) \subset T(v)$.
We can assume by induction that
$P(u_i) \prec T(v_j)$ is known
for all $u_i \in Chd(u)$
and for all
$v_j \in V(T(v)) - \{ v \}$.
Let $M = \{ (u_i,v_j) \,|\, P(u_i) \prec T(v_j) ~\land~ v_j \in V(T(v)) \}$.
We define $OCC(u_i,M)$ and $Occ(u_i,M)$ by
\begin{eqnarray*}
OCC(u_i,M) & = & \{(u_i,v_j) \,|\, (u_i,v_j) \in M\}.\\
Occ(u_i,M) & = & |OCC(u_i,M)|.
\end{eqnarray*}
See Fig.~\ref{fig:d2d3} for an illustration.
A node $u_i$ with $Occ(u_i,M)=h$ is called a node of \emph{rank} $h$.
Note that $u_i$, $v_j$, and $M$ appearing above depend on $(u,v)$.

The crucial task is to find an injective mapping $\psi$
(called a \emph{valid mapping})
from $P(u)$ to $V(T(v))-\{v\}$ such that
$P(u_i) \prec T(\psi(u_i))$ holds for all $u_i$ ($i=1,\ldots,d$) and
there is no ancestor/descendant relationship between
any $\psi(u_i)$ and $\psi(u_j)$ ($u_i \neq u_j$).
If this task can be performed in $O(f(d,m,n))$ time,
from Proposition~\ref{prop:bottom-up},
the total time complexity
will be $O^{\ast}(f(d,m,n))$.
We assume w.l.o.g. that $\psi$ is given as a set of mapping pairs.

Hereafter, we let $Chd(u)=\{u_{i_1},\ldots,u_{i_d}\}$.
Since we consider the case of $occ(P,T)$ $=2$,
we assume w.l.o.g. that all $u_{i_k}$s have rank 2
(i.e., $Occ(u_{i_k},M)=2$ for $k=1,\ldots,d$).
Accordingly, 
we let $OCC(u_{i_k},M)=\{(u_{i_k},v_{j_{k,0}}),(u_{i_k},v_{j_{k,1}})\}$
for $k=1,\ldots,$ $d$.
As in \cite{aspvall1979},
we construct a directed graph $G_2(V_2,E_2)$ by
\begin{eqnarray*}
V_2 & = & \{u_{i_{k,0}},u_{i_{k,1}} \mid u_{i_k} \in Chd(u)\},\\
E_2 & = & \{(u_{i_{k,p}},u_{i_{h,q}}) \mid
v_{j_{k,p}} \in AncDes(v_{j_{h,1-q}}),~h \neq k \},
\end{eqnarray*}
where $u_{i_{k,p}}$s are newly introduced symbols.
See also Fig.~\ref{fig:occ-2}.
Intuitively, an arc $(u_{i_{k,p}},u_{i_{h,q}})$ implies that if
$(u_{i_k},v_{j_{k,p}})$ is in the inclusion mapping then it is possible for
$(u_{i_h},v_{j_{h,q}})$, but not $(u_{i_h},v_{j_{h,1-q}})$, to be in the
mapping, too.

\begin{figure}[t!]
\centering
\includegraphics[width=12cm]{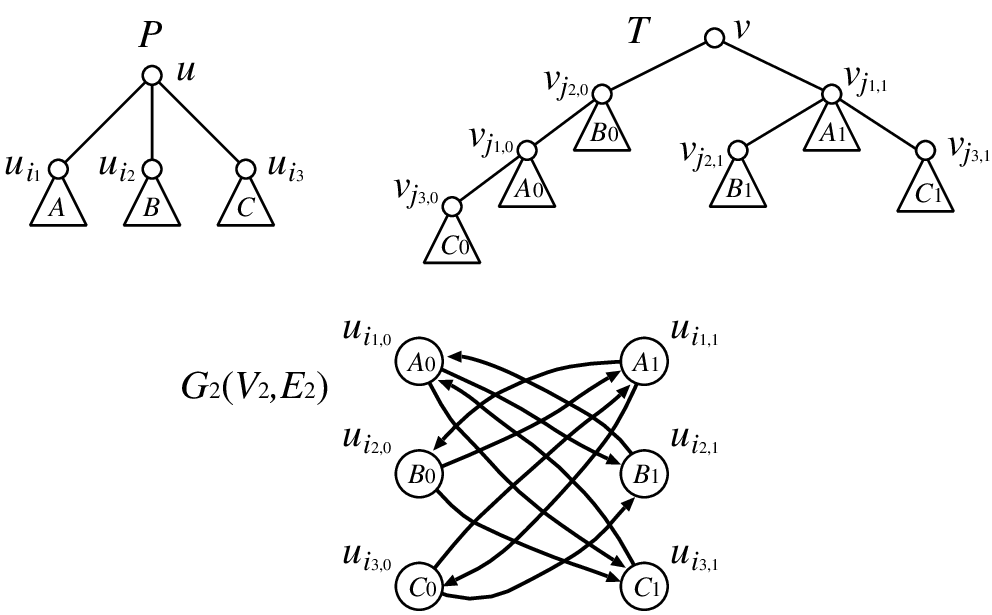}
\caption{Example of
$G_2(V_2,E_2)$
constructed from $P(u)$ and $T(v)$ in the case of $occ(P,T)=2$.
Each vertex $u_{i_{k,p}}$ is represented by the corresponding subtree in $T$,
where triangles labeled by the same capital letter represent
isomorphic subtrees.
$G_2(V_2,E_2)$ has two strongly connected components
$\{u_{i_{1,0}},u_{i_{2,1}},u_{i_{3,1}}\}$ and
$\{u_{i_{1,1}},u_{i_{2,0}},u_{i_{3,0}}\}$,
whereas there is only one consistent assignment
$\{u_{i_{1,0}}=1,u_{i_{2,1}}=1,u_{i_{3,1}}=1,
u_{i_{1,1}}=0,u_{i_{2,0}}=0,u_{i_{3,0}}=0\}$,
which corresponds to a mapping
$\psi=\{(u_{i_1},v_{j_{1,0}}),(u_{i_2},v_{j_{2,1}}),(u_{i_3},v_{j_{3,1}})\}$.}
\label{fig:occ-2}
\end{figure}

\begin{proposition}
There exists a path (resp., an edge) from 
$u_{i_{k,p}}$ to $u_{i_{h,q}}$
if and only if
there exists a path (resp., an edge) from 
$u_{i_{h,1-q}}$ to $u_{i_{k,1-p}}$
\end{proposition}
\begin{proof}
It is shown in \cite{aspvall1979} that $G_2(V_2,E_2)$ has a duality property:
$G_2$ is isomorphic to the graph obtained from $G_2$ by reversing the direction
of all the edges and complementing the names of all vertices.
Since $u_{i_{h,q}}$ and $u_{i_{h,1-q}}$
(resp., $u_{i_{k,p}}$ and $u_{i_{k,1-p}}$) correspond to complementary
variables,
the proposition holds.
\end{proof}

Consider a 0-1 assignment to $V_2$, 
where 0 and 1 correspond to {\bf false} and {\bf true}, respectively.
An assignment is called \emph{consistent} if
the following conditions are satisfied.
\begin{itemize}
\item $u_{i_{k,0}} + u_{i_{k,1}} = 1$ holds for all $k=1,\ldots,d$,
\item if $u_{i_{k,p}}=1$,
all vertices reachable from $u_{i_{k,p}}$ have value 1.
\end{itemize}
Note that the first condition implies that
$u_{i_{k,1-p}}$ corresponds to the negation of $u_{i_{k,p}}$,
which further means that
$u_{i_k}$ must be mapped to exactly one of $v_{j_{k,0}}$ and
$v_{j_{k,1}}$.
Note also that the second condition implies that if $u_{i_{k,p}}=0$,
all vertices reachable to $u_{i_{k,p}}$ have value 0.

\begin{proposition}
$P(u) \subset T(v)$ holds if and only if
there exists a consistent assignment.
Furthermore, $\psi$ can be obtained from the vertices
to which 1 is assigned.
\end{proposition}
\begin{proof}
Suppose that there exists
a consistent assignment.
Then, we can construct an inclusion mapping $\psi$ 
for $Chd(u)$ by letting
$\psi(u_{i_k})=v_{j_{k.p}}$ for $p$ such that $u_{i_{k,p}}=1$,
for all $u_{i_k} \in Chd(u)$, where
the validity
follows from the above two conditions 
and the meaning of an arc.

Conversely, suppose that there exists an inclusion mapping $\psi$.
Then, we let $u_{i_{k,p}}=1$ if and only if
$\psi(u_{i_k})=v_{j_{k.p}}$ for all $u_{i_k} \in Chd(u)$,
which clearly satisfies the above two conditions.
\end{proof}

As in \cite{aspvall1979},
we have the following proposition.

\begin{proposition}
There exists a consistent assignment to $V_2$ if and only if
there is no $k$ such that
$u_{i_{k,0}}$ and $u_{i_{k,1}}$ belong to the same strongly
connected component in $G_2(V_2,E_2)$.
\end{proposition}

The strongly
connected components can be computed in linear time \cite{tarjan1972}.
Furthermore, a consistent assignment can be obtained by 
greedily assigning 1 to vertices from deeper to shallower SCCs
under the DFS (depth first search) ordering as in \cite{aspvall1979}.
Since this procedure can clearly be done in polynomial time,
the following theorem holds.

\begin{theorem}
Unordered tree inclusion can be solved in polynomial time
if $occ(P,$ $T)=2$.
\label{thm:poly}
\end{theorem}

\section{An $O^{\ast}(1.619^d)$-time algorithm for the case of $occ(P,T)=3$}
\label{sec:occ3}

In this section, we present an $O^{\ast}(1.619^d)$-time algorithm
for the case of $occ(P,T)=3$, where $d$ is the maximum degree of $P$,
$m=|V(P)|$, and $n=|V(T)|$.
Note that this case remains NP-hard from Theorem~\ref{thm:unique_leaves_NP-complete}.

The basic strategy is to
combine bottom-up dynamic programming and detection of a consistent assignment
as in Section~\ref{sec:poly} to determine whether $P(u) \subset T(v)$ holds,
where a recursive procedure is employed here for finding a consistent assignment.
Let $Chd(u)=\{u_{i_1},\ldots,u_{i_d}\}$.
As in Section~\ref{sec:poly},
we can assume that $P(u_{i_k}) \prec T(v_{j_h})$ is known for all $u_{i_k}$
and for all $v_{j_h} \in V(T(v)) - \{ v \}$,
and we let
$M = \{ (u_{i_k},v_{j_h}) \,|\, P(u_{i_k}) \prec T(v_{j_h}) ~\land~ v_{j_h} \in V(T(v)) \}$.

Let $d_3$ (resp., $d_2$) be the number of $u_{i_k}$s of rank 3
(resp., rank 2) (see also Fig.~\ref{fig:d2d3}).
We assume w.l.o.g. that $d_2 + d_3 = d$
because $Occ(u_{i_k},M)=1$ means that $\psi(u_{i_k})$ is uniquely determined
and thus we can ignore $u_{i_k}$s with $Occ(u_{i_k},M)=1$.

\begin{figure}[t!]
\centering
\includegraphics[width=12cm]{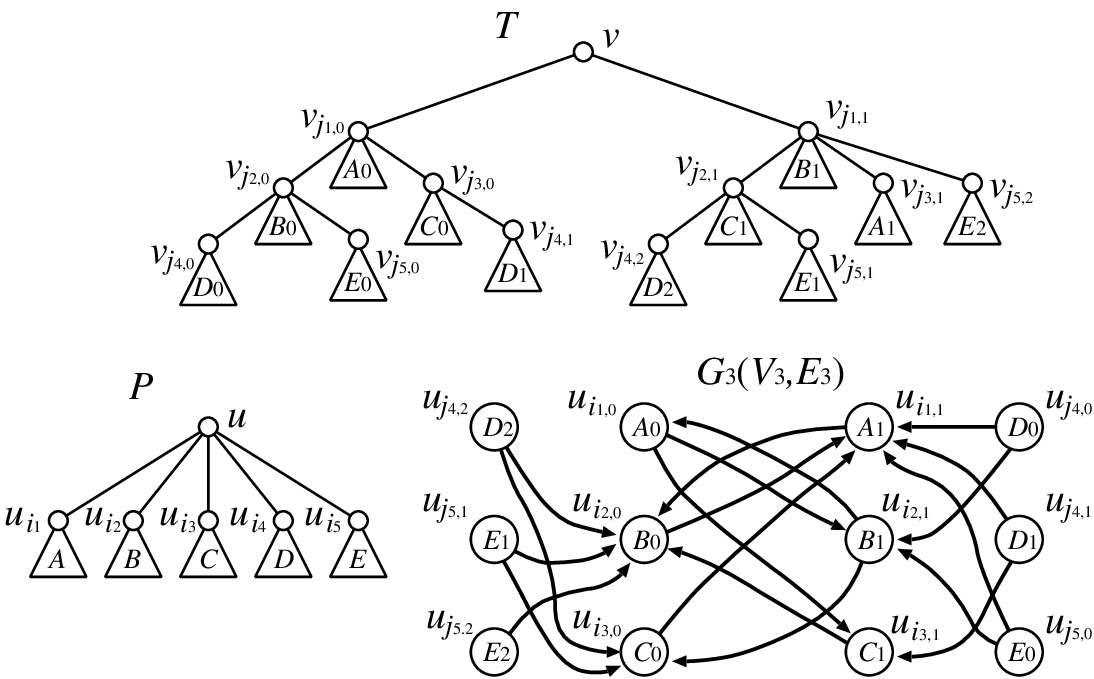}
\caption{Example of
$G_3(V_3,E_3)$
constructed from $P(u)$ and $T(v)$ in the case of $occ(P,T)=3$.
$u_{j_{4,0}}$ and $u_{j_{5,0}}$ are inadmissible vertices,
and $(u_{j_{4,1}},u_{j_{5,1}})$ is an inadmissible pair,}
\label{fig:occ-3}
\end{figure}

We construct $G_2(V_2,E_2)$ as in Section~\ref{sec:poly},
using only $u_{i_k}$s with rank 2
and the corresponding $v_{j_h}$s, considering
ancestor-descendant relations only among them.
Then,
for each $u_{i_k} \in Chd(u)$ such that 
$OCC(u_{i_k},M) = \{(u_{i_k},v_{j_{k,0}}),(u_{i_k},v_{j_{k,1}}),$
$(u_{i_k},v_{j_{k,2}})\}$,
we let $V_{OCC3}(u_{i_k}) = \{u_{i_{k,0}},u_{i_{k,1}},u_{i_{k,2}} \}$,
where $u_{i_{k,p}}$s are newly introduced symbols.
Let $V_{OCC3} = \bigcup_{Occ(u_{i_k},M)=3} V_{OCC3}(u_{i_k})$.
Then, we construct $G_3(V_3,E_3)$ from $G_2(V_2,E_2)$ by
\begin{eqnarray*}
V_3 & = & V_2 \cup V_{OCC3},\\
E_3 & = & E_2 \cup \{ (u_{i_{k,p}},u_{i_{h,q}}) \mid 
u_{i_{k,p}} \in V_{OCC3}, 
u_{i_{h,q}} \in V_2,
v_{i_{h,1-q}} \in AncDes(v_{i_{k,p}}) \}.
\end{eqnarray*}
See Fig.~\ref{fig:occ-3} for an example of $G_3(V_3,E_3)$.

\begin{definition}
\label{def:inad}
We say that $u_{i_{k,p}} \in V_{OCC3}$ is an \emph{inadmissible vertex}
if there exist paths from $u_{i_{k,p}}$ to $u_{i_{l,0}}$ and $u_{i_{l,1}}$
in $G_3(V_3,E_3)$ for some $u_{i_l} \in Chd(u)$ of rank 2.
We also say that $(u_{i_{k,p}},u_{i_{h,q}}) \in V_{OCC3} \times V_{OCC3}$ 
($k \neq h$)
is an \emph{inadmissible pair}
if $v_{i_{h,q}} \in AncDes(v_{i_{k,p}})$ holds, or
there exist a path reachable from $u_{i_{k,p}}$ to $u_{i_{l,0}}$ in $G_3(V_3,E_3)$
and a path reachable from $u_{i_{h,q}}$ to $u_{i_{l,1}}$ in $G_3(V_3,E_3)$
for some $u_{i_l} \in Chd(u)$ of rank 2.
\end{definition}

It is to be noted that
an inadmissible vertex or
an inadmissible pair $(u_{i_{k,p}},u_{i_{h,q}})$ 
cannot appear in any injective mapping $\psi$
for $P(u) \subset T(v)$ because the use of an inadmissible vertex or
an inadmissible pair would make a consistent assignment impossible.
Accordingly, we can assume w.l.o.g. that there does not exist 
an inadmissible vertex $u_{i_{k,p}}$ in $V_{OCC3}$.

\begin{proposition}
Suppose that there exists a consistent assignment on vertices in $G_2(V_2,E_2)$
in the sense defined in Section~\ref{sec:poly}.
If there does not exist an inadmissible pair,
there exists a valid mapping $\psi$.
Furthermore, such a mapping can be found in polynomial time.
\label{prop:admissible}
\end{proposition}
\begin{proof}
We present a greedy algorithm for finding a consistent assignment,
from which a valid mapping can be obtained.
Beginning with an empty assignment on all vertices in $V_2$,
we repeat the following procedure in any order:
for each $u_{i_{k}}$ of rank 3,
assign 1 to $u_{i_{k,0}}$,
assign 0 to $u_{i_{k,1}}$ and $u_{i_{k,2}}$,
and assign 1 to all vertices in $G_3(V_3,E_3)$ reachable from $u_{i_{k,0}}$.
Finally, we extend the resulting assignment to a consistent assignment
by assigning 1 to remaining vertices from deeper to shallower strongly
connected components under the DFS ordering.
Clearly, this algorithm works in polynomial time.
It is also seen from the definition of the inadmissible pair that
this algorithm always finds a consistent assignment. 
\end{proof}

We denote the procedure in the proof of Proposition~\ref{prop:admissible}
by $FindMappingAD(M)$.
This procedure returns {\bf true} or {\bf false}.
{\bf true} corresponds to the case where a consistent assignment and
a valid mapping $\psi$ exist.
It is straightforward to modify the procedure so that it outputs
$\psi$ when it exists.

In order to handle inadmissible pairs,
we employ a simple recursive procedure.
Suppose that $(u_{i_{k,p}},u_{i_{h,q}})$ is an inadmissible pair.
If we include $(u_{i_k},v_{j_{k,p}})$ in $\psi$,
we cannot include $(u_{i_h},v_{j_{h,q}})$ in $\psi$.
In this case, $d_3$ is decreased by 2.
If we do not include $(u_{i_k},v_{j_{k,p}})$,
we can delete this pair from $M$, which decreases $d_3$ by 1.
Based on this idea, we obtain the following main procedure for the case of
$occ(P,T)=3$.
Note that if we include $(u_{i_k},v_{j_{k,p}})$ in $\psi$,
all pairs 
$(u_{i_k},v_{j_{k,r}})$ with $r=0,1,2$ are removed from $M$.
Furthermore,
all pairs $(u_{i_h},v_{j_{h,q}})$ such that
$v_{j_{h,q}} \in AncDes(v_{j_{k,p}})$ are removed from $M$,
which may cause further removal.
$Update(M,(u_{i_k},v_{j_{k,p}}))$ executes this updating procedure
while
making the corresponding 0-1 assignments on $G_3(V_3,E_3)$.

\begin{algorithm}
\caption{$FindMapping(M)$}
\label{alg:findmap}
\begin{algorithmic}
\IF{there exists an inadmissible pair $((u_{i_k},v_{j_{k,p}}),(u_{i_h},v_{j_{h,q}}))$}
\STATE{$M_1 := Update(M,(u_{i_k},v_{j_{k,p}}))$}
\STATE{$M_2 := M - \{(u_{i_k},v_{j_{k,p}}) \}$}
\IF{$FindMapping(M_1) = \TRUE$}
\RETURN{\TRUE}
\ELSE
\RETURN{$FindMapping(M_2)$}
\ENDIF
\ENDIF
\RETURN{$FindMappingAD(M)$}
\end{algorithmic}
\end{algorithm}

\begin{theorem}
Unordered tree inclusion can be solved in $O^{\ast}(1.619^d)$ time
if $occ(P,$ $T)=3$.
\label{thm:occ-3}
\end{theorem}
\begin{proof}
It follows from the discussions above that 
$FindMapping(M)$ correctly decides whether $P(u) \subset T(v)$
(when $u$ and $v$ have the same label).
Therefore, we analyze the exponential factor (depending on $d$)
of the time complexity of $FindMapping(M)$.

Let $f(k)$
denote
the number of times that
$FindMapping(M)$ is called when $k = |\{ u_i \,|\, Occ(u_i,M) = 3 \}|$.
Clearly, if $k \leq 1$, $f(k) \leq 1$.
Otherwise (i.e., $k \geq 2$), it may invoke two recursive calls:
one with at most $k-2$ nodes of rank 3 and the other with at most $k-1$ nodes
of rank 3.
Therefore, we have
\begin{eqnarray*}
f(k) & \leq &  f(k-1) + f(k-2),
\end{eqnarray*}
from which $f(k) = O(1.619^k)$ follows (c.f., Fibonacci number).

Since $d_3 \leq d$ holds and
both $FindMappingAD(M)$ and $Update(M,(u_{i_k},u_{j_{k,p}}))$
work in polynomial time per execution,
the total time complexity
is $O^{\ast}(1.619^{d})$.
\end{proof}

\section{A randomized algorithm for the case of $h(P)=1$ and $h(T)=2$}
\label{low-height}

Finally, we consider the case of $h(P)=1$ and $h(T)=2$,
denoted by {\bf IncH2}.
This problem variant is NP-hard according to
Corollary~\ref{corollary:KM_hardness}.
We assume w.l.o.g. that the roots of $P$ and $T$ have the same
unique label and thus they must match in any inclusion mapping.

Let $U=\{u_1,\ldots,u_d\}$ be the set of children of $r(P)$.
Let $v_1,\ldots,v_g$ be the children of $r(T)$, and
let $v_{i,1},\ldots,v_{i,n_i}$ be the children of each $v_i$.

First, we assume that $\ell(u_i) \neq \ell(u_j)$ holds for
all $i \neq j$,
where $\ell(v)$ denotes the label of $v$.
This special case is denoted by {\bf IncH2U}.
Recall that {\bf IncH2U} remains NP-hard
from the condition of $occ(P)=1$ of Corollary~\ref{corollary:KM_hardness}.

{\bf IncH2U} can be solved by a reduction to CNF SAT,
different from the one mentioned in Section~\ref{sec:poly}.
(In fact, it can be considered as an inverse reduction of the one originally
used to prove the NP-hardness of unordered tree inclusion
by Kilpel\"{a}inen and Mannila~\cite{kilpelainen1995ordered}.)
For each $u_i$, we define $X^{POS}_i$ and $X^{NEG}_i$ by
\begin{eqnarray*}
X^{POS}_i & = & \{ x_j \,|\, \ell(u_i)=\ell(v_j) \},\\
X^{NEG}_i & = & \{ x_j \,|\, (\exists v_{j,k} \in  Chd(v_j))(\ell(u_i)=\ell(v_{j,k})) \}.
\end{eqnarray*}
For each $u_i$, we construct a clause $C_i$ by
$$
C_i = \left( \bigvee_{x_j \in X^{POS}_i} x_j \right)
\lor \left( \bigvee_{x_j \in X^{NEG}_i} \lnon{x_j} \right).
$$
Then, the resulting SAT instance is $\{C_1,\ldots,C_d\}$.
Intuitively, $x_j=1$ corresponds to the case where $u_i$ is mapped to $v_j$,
where $\ell(u_i)=\ell(v_j)$.
Of course, multiple $v_j$s may correspond to $u_i$.
However, it is enough to consider an arbitrary one.

\begin{proposition}
{\bf IncH2U} can be solved in $O^{\ast}(1.234^d)$ time.
\end{proposition}
\begin{proof}
First we prove the correctness of the reduction, where
we assume w.l.o.g. that $r(P)$ is mapped to $r(T)$.
Suppose that there exists an inclusion mapping $\phi$ from $V(P)$ to
$V(T)$.
Then, we let $x_j=1$ if $\phi(u_i)=v_j$, and $x_j=0$ if $\phi(u_i)=v_{j,k}$.
An arbitrary assignment can be done on each of the other variables.
Then, we can see that
there is no inconsistency on the resulting assignment and
all $C_i$s are satisfied.
Conversely, suppose that there exists a satisfying assignment on $C_i$s.
We let $\phi(u_i)=v_j$ if $x_j=1$ and $\ell(u_i)=\ell(v_j)$.
Otherwise, we can let $\phi(u_i)=v_{j,k}$ for some $v_j$ such that $x_j=0$
and $\ell(u_i)=\ell(v_{j,k})$.
This $\phi$ gives an inclusion mapping.

Next we consider the time complexity.
In order to solve the satisfiability instance, we use
Yamamoto's $O^{\ast}(1.234^{d})$-time algorithm for SAT with
$d$~clauses~\cite{yamamoto2005sat}.
Since the other parts can be done in polynomial time,
we have the proposition.
\end{proof}

In order to solve {\bf IncH2},
we combine two algorithms:
(A1)~a random sampling-based algorithm;
and (A2)~a modified version of the $O(d 2^d mn^2)$-time algorithm
in Section~\ref{sec:improved}.

For (A1), we employ the \emph{color-coding} technique \cite{alon1995color}.
Let $d_0$ be the number of $u_i$s having unique labels, and
let $d_1 \leq d_2 \leq \cdots \leq d_h$ be the multiplicities of
the other labels in $U$.
Define $\alpha = 1 - \frac{d_0}{d}$.
Note that $d_0 + d_1 + \cdots + d_h = d$ and $d - d_0 = \alpha d$ hold.

For each label $a_i$ with $d_i \geq 2$ (i.e., $i>0$),
we relabel the nodes in $P$ having label $a_i$ by
$a^1_i,a^2_i,\ldots,a^{d_i}_i$ in an arbitrary order.
For each node $v$ in $T$ having label $a_i$,
we assign $a^j_i$ ($j=1,\ldots,d_i$) to $v$ uniformly at random, and then
apply the SAT-based algorithm for {\bf IncH2U}.
Let $M$ be the set of pairs in an inclusion mapping from $P$ to $T$.
If all nodes of $T$ appearing in $M$ have different labels,
a valid inclusion mapping can be obtained.
This success probability is given by
\begin{eqnarray*}
{\frac {d_1 !}{d_1^{d_1}}} \cdot 
{\frac {d_2 !}{d_2^{d_2}}} \cdots 
{\frac {d_h !}{d_h^{d_h}}} ~\geq~ {\frac {(\alpha d)!}{(\alpha d)^{(\alpha d)}}} .
\end{eqnarray*}
This inequality can be proved by repeatedly applying
$$
{\frac {d_1 !}{d_1^{d_1}}} \cdot 
{\frac {d_2 !}{d_2^{d_2}}} ~\geq~
{\frac {(d_1+d_2)!}{(d_1+d_2)^{d_1+d_2}}},
$$
which is seen from
$$
{\frac {(d_1+d_2)^{d_1+d_2}}{d_1^{d_1} d_2^{d_2}}}
~\geq~
\left(
\begin{array}{c}
d_1+d_2\\
d_1
\end{array}
\right) ~=~
{\frac {(d_1+d_2)!}{d_1! d_2!}}.
$$
Since ${\frac {k!}{k^k}} \geq e^{-k}$ holds for sufficiently large $k$,
the success probability is at least $e^{-\alpha d}$.
Therefore, if we repeat the random sampling procedure $e^{\alpha d}$ times,
the failure probability is at most
$(1-e^{-{\alpha d}})^{e^{\alpha d}} \leq e^{-1} < {\frac 1 2}$ 
because $\ln\left[ (1-{\frac 1 x})^x \right] = x\ln(1-{\frac 1 x})
\leq x (-{\frac 1 x}) = -1 = \ln(e^{-1})$ holds for any $x > 1$.

If we repeat the procedure $k (\log n) e^{\alpha d}$ times where $k$ is
any positive constant (i.e., the total time complexity is
$O^{\ast}(1.234^d \cdot e^{\alpha d})$),
the failure probability is at most ${\frac 1 {n^k}}$.

For (A2),
we modify the $O(d 2^d mn^2)$-time algorithm as follows.
Recall that if there exist labels with multiplicity more than one,
$S(v,v_i)$ is a multi-set.
In order to represent a multi-set, we memorize the multiplicity of each label.
Then, the number of distinct multi-sets is given by
\begin{eqnarray*}
N(d_0,\ldots,d_h) & = & 2^{d_0} \cdot  \prod_{l=1}^h (d_l+1) .
\end{eqnarray*}
Since $d_i+1 \leq 3^{\lceil d_i/2 \rceil}$ holds for any $d_i \geq 2$,
this number is bounded as follows:
\begin{eqnarray*}
N(d_0,\ldots,d_h) & \leq & 2^{d_0} \cdot 3^{\lceil (d-d_0)/2 \rceil}.
\end{eqnarray*}
Then, the time complexity of (A2) is $O^{\ast}(2^{(1-\alpha)d} \cdot 3^{(\alpha/2) d})$.

Since we can
select
the minimum of the time complexities of (A1) and (A2),
the resulting time complexity is given by
\begin{eqnarray*}
\max_{\alpha} \min(O^{\ast}(1.234^d \cdot e^{\alpha d}),O^{\ast}(2^{(1-\alpha)d} \cdot 3^{(\alpha/2) d})) .
\end{eqnarray*}
Since $1.234^d \cdot e^{\alpha d}$ and 
$2^{(1-\alpha)d} \cdot 3^{(\alpha/2)d}$
are increasing and decreasing functions of $\alpha$, respectively,
this maximum is attained when 
$1.234 \cdot e^{\alpha} = 2^{(1-\alpha)} \cdot 3^{(\alpha/2)}$.
By numerical calculation, we have $\alpha \approx 0.42217$, from which
the following theorem follows.

\begin{theorem}
{\bf IncH2} can be solved in randomized $O^{\ast}(1.883^d)$ time
with probability
at least $1-{\frac 1 {n^k}}$, where $k$ is any positive constant.
\label{thm:low-height}
\end{theorem}

The above algorithm may be derandomized by using $k$-perfect hash families as
in~\cite{alon1995color}.
However, since the construction of a $k$-perfect hash family has
a high complexity, the resulting algorithm would have a time complexity much
worse than $O^{\ast}(2^d)$.

\section{Concluding remarks}

We have presented a new algorithm for unordered tree inclusion running in
$O^{\ast}(2^{d})$ time, thus reducing the exponent~$2d$ in the previously
best known bound on the time complexity~\cite{kilpelainen1995ordered} to~$d$.
However, the $2^{d}$-factor may not be optimal.
For example, our randomized algorithm for the special case of $h(P) = 1$ and
$h(T) = 2$ runs in $O^{\ast}(1.883^d)$ time, which suggests that further
improvements could be possible.
However, we were unable to obtain an $O^{\ast}((2-\varepsilon)^d)$-time algorithm
for any constant $\varepsilon > 0$, even when $h(P) = h(T) = 2$.
Similarly, we could not obtain an $O^{\ast}((2-\varepsilon)^d)$-time algorithm for
any constant $\varepsilon > 0$ when $occ(P,T) = 4$.
Therefore, to develop an $O^{\ast}((2-\varepsilon)^d)$-time algorithm for
unordered tree inclusion or to prove an $\Omega(2^d)$ lower bound
(e.g., using recent techniques
from~\cite{abbound2016subtree,abbound2014alignment,bringmann2018editlb} for
proving lower bounds on various tree and sequence matching problems)
is left as an open problem.

Future work includes generalizing our techniques and applying them to
the \emph{extended tree inclusion problem} mentioned in
Section~\ref{sec:applications}.
This problem variant
was introduced by Mori et al.~\cite{MTJHTA_15}
as a way to make
unordered tree inclusion more useful for practical pattern matching
applications.
It asks for an optimal connected subgraph of~$T$ (if any) that can be obtained
by applying node insertion operations as well as node relabeling operations
to~$P$ while allowing non-uniform costs to be assigned to the different
node operations.
It was shown in \cite{MTJHTA_15}
that the unrooted case can be solved in
$O^{\ast}(2^{2d})$ time, and a further extension of the problem that also
allows at most~$K$ node deletion operations can be solved in
$O^{\ast}((ed)^K K^{1/2} 2^{2(dK+d-K)})$ time, where $e$ is the base of
the natural logarithm.

\bibliographystyle{plain}
\bibliography{impinctcsbib}

\end{document}